\documentclass[conference]{IEEEtran}

\pdfoutput=1

\usepackage{amsfonts}
\usepackage{cite} 
\usepackage{amsmath,amsthm} 
\usepackage{amssymb} 
\usepackage{epsfig}
\usepackage{graphicx}
\usepackage{etoolbox}

 \let\bbordermatrix\bordermatrix
\patchcmd{\bbordermatrix}{8.75}{4.75}{}{}
\patchcmd{\bbordermatrix}{\left(}{\left[}{}{}
\patchcmd{\bbordermatrix}{\right)}{\right]}{}{}

\newtheorem{theorem}{Theorem}

\newtheorem{definition}{Definition}
\newtheorem{corollary}{Corollary}

\newtheorem{remark}{Remark}

\newcommand{\sr}{\stackrel}

\newcommand{\tri}{\sr{\triangle}{=}}

\newcommand{\noi}{\noindent}

\newcommand{\be}{\begin{equation}}
\newcommand{\ee}{\end{equation}}
\newcommand{\bea}{\begin{eqnarray}}
\newcommand{\eea}{\end{eqnarray}}
\newcommand{\bes}{\begin{eqnarray*}}
\newcommand{\ees}{\end{eqnarray*}}

\newcommand{\bfi}{\begin{figure}}
\newcommand{\bfit}{\begin{figure}[t]}
\newcommand{\bfib}{\begin{figure}[b]}
\newcommand{\bfih}{\begin{figure}[h]}
\newcommand{\bfip}{\begin{figure}[p]}
\newcommand{\efi}{\end{figure}}

\newcommand{\bi}{\begin{itemize}}
\newcommand{\ei}{\end{itemize}}
\newcommand{\ben}{\begin{enumerate}}
\newcommand{\een}{\end{enumerate}}

\renewcommand{\IEEEQED}{\IEEEQEDclosed}

\begin{document}

\sloppy

\title{Applications of Information Nonanticipative Rate Distortion Function}

\author{
  \IEEEauthorblockN{Photios A. Stavrou, Christos K. Kourtellaris, and Charalambos D. Charalambous}
  \IEEEauthorblockA{ECE Department, University of Cyprus, Nicosia, Cyprus\\
    {\it Email:\{stavrou.fotios,kourtellaris.christos,chadcha\}@ucy.ac.cy}} 
}

\maketitle

\begin{abstract}
The objective of this paper is to further investigate  various applications of information Nonanticipative Rate Distortion Function (NRDF) by discussing  two working examples, the Binary Symmetric Markov Source with parameter $p$ (BSMS($p$)) with Hamming distance distortion, and the multidimensional partially observed Gaussian-Markov source. For the BSMS($p$), we give the solution to the NRDF, and we use it to compute the Rate Loss (RL) of causal codes with respect to noncausal codes. For the multidimensional Gaussian-Markov source, we give the solution to the NRDF, we show  its operational meaning via joint source-channel matching over a vector of parallel Gaussian channels, and we compute the RL of causal and zero-delay codes with respect to noncausal codes.   
\end{abstract}


\section{Introduction}\label{introduction}
\par In this paper, we consider an information theoretic measure called Nonanticipative Rate Distortion Function (NRDF) \cite{stavrou-charalambous2013isit,charalambous-stavrou-ahmed2014a} which is a variation of the classical RDF \cite{berger}, and we discuss some of its applications in problems on information theory.  In \cite{stavrou-charalambous2013isit}, it is pointed out that the information NRDF and nonanticipatory $\epsilon$-entropy introduced in \cite{gorbunov-pinsker} to facilitate real-time applications are equivalent notions, and a variational equality is derived and utilized to introduce a Blahut-Arimoto Algorithm (BAA) to iteratively compute the information NRDF. In addition, existence of the optimal nonanticipative reproduction conditional distribution  is shown, under the topology of weak convergence of probability measures, while in \cite{charalambous-stavrou-ahmed2014a}, the closed form expression of the optimal reproduction conditional distribution for stationary processes is derived. Moreover, in \cite{charalambous-stavrou-ahmed2014a}, the realization of the optimal reproduction distribution of the information NRDF  is discussed (see Fig.~\ref{discrete_time_communication_system}) in the context of filtering applications with  fidelity constraints.\\ 
In this paper, we present results in the following directions.
\vspace*{0.2cm}\\
\noi{\bf (R1)} Compute  the NRDF in closed form for two examples of sources with memory: (a) the {\it Binary Symmetric Markov Source with parameter $p$ with Hamming distortion} (BSMS($p$)), for which the classical RDF is only  known in the distortion region $0\leq{D}\leq{D}_c$ \cite{gray1971}, while for the rest $D_c\leq{D}\leq\frac{1}{2}$ only  upper and lower  bounds are known \cite{berger1977}. We show that the solution of the NRDF is a tight upper bound for $D_c\leq{D}\leq\frac{1}{2}$, and performs much more reliably in comparison to the upper bound found in \cite{berger1977}; (b) the {\it multidimensional Gaussian-Markov source}, for which only upper bounds are known, since no closed form expression is given in the literature apart from the first-order (scalar) Gauss-Markov sources \cite[Th. 3]{derpich-ostergaard2012}.\\
\noi{\bf (R2)} Compute the Rate Loss (RL) of causal codes, that is, the gap between the Optimal Performance Theoretically Attainable (OPTA) by causal codes with respect to noncausal codes for the BSMS($p$).\\
\noi{\bf (R3)} Compute the RL of causal and zero-delay codes with respect to noncausal codes for the multidimensional Gaussian-Markov source, and show achievability of the NRDF using symbol-by-symbol transmission \cite{gastpar2003}.\\
\noi{\bf (R4)} Provide an alternative characterization of the closed form expression to the information NRDF, from which a lower bound on the NRDF similar to Shannon's Lower Bound (SLB) \cite[Ch. 4]{berger} can be derived, for any source with memory, including Gaussian-Markov sources. This bound is  utilized in the derivation of the closed form expression of the multidimensional Gaussian-Markov source.
\vspace*{0.2cm}\\
\noi The alternative characterization of the solution  to the information NRDF (see Theorem~\ref{alternative_expression}) is the analogue of the single letter characterization of  the classical RDF of  discrete memoryless  sources, often used to  facilitate the computation of the classical BAA  \cite[Th. 6.3.9]{blahut1987}.\\
\noi Finally, we point out that the multidimensional Gaussian-Markov source example is a generalization to arbitrary dimensions of the example considered in \cite[Cor. 1.2]{ma-ishwar2011} for systems with low delay tolerance at both the encoder and decoder, such as, the classical Differential Predictive Coded Modulation (DPCM) system\cite{berger}, often applied to compression applications  of video, audio, image, and speech coding.


\section{NRDF on Abstract Spaces}\label{nonanticipative_rdf}
\par In this section, we define the information NRDF by adopting the general mathematical framework described in \cite{stavrou-charalambous2013isit}. \\
\noi {\bf Notation.} Let $\mathbb{N} \tri \{0,1,\ldots\}$. Introduce two sequence of spaces $\{({\cal X}_n,{\cal B}({\cal X }_n)):n\in\mathbb{N}\}$ and $\{({\cal Y}_n,{\cal B}({\cal Y}_n)):n\in\mathbb{N}\},$ where ${\cal X}_n,{\cal Y}_n, n\in\mathbb{N}$, are Polish spaces, and ${\cal B}({\cal X}_n)$ and ${\cal B}({\cal Y}_n)$ are Borel $\sigma-$algebras of subsets of ${\cal X}_n$ and ${\cal Y}_n$, respectively. Points in ${\cal X}^{\mathbb{N}}\tri{{\times}_{n\in\mathbb{N}}}{\cal X}_n$ are denoted by ${\bf x}\tri\{x_0,x_1,\ldots\}\in{\cal X}^{\mathbb{N}}$, while their restrictions to finite coordinates are denoted by $x^n\tri\{x_0,x_1,\ldots,x_n\}\in{\cal X}_{0,n},$ for $n\in\mathbb{N}$, and similarly of ${\cal Y}_n$. Let ${\cal B}({\cal X}^{\mathbb{N}})\tri\odot_{i\in\mathbb{N}}{\cal B}({\cal X}_i)$ denote the $\sigma-$algebra on ${\cal X}^{\mathbb{N}}$ generated by cylinder sets and similarly for ${\cal B}({\cal Y}^{\mathbb{N}})\tri\odot_{i\in\mathbb{N}}{\cal B}({\cal Y}_i)$, while ${\cal B}({\cal X}_{0,n})$ and ${\cal B}({\cal Y}_{0,n})$ denote the $\sigma-$algebras with bases over $A_i\in{\cal B}({\cal X}_i)$, and $B_i\in{\cal B}({\cal Y}_i),~i=0,1,\ldots,n$, respectively. Let ${\cal Q}({\cal Y};{\cal X})$ denote the set of stochastic kernels on ${\cal Y}$ given ${\cal X}$ and ${\cal M}({\cal X})$ the set of probability measures on ${\cal X}$.\\
\noi{\bf Source Distribution.} Consider the sequence of source distributions $\{p_n(dx_n;x^{n-1}):n\in\mathbb{N}\}$, where $p_n(\cdot;\cdot)\in{\cal Q}({\cal X}_n;{\cal X}_{0,n-1})$. For $A\in{\cal B}({\cal X}_{0,n})$ a cylinder set of the form $
A\tri\big\{{\bf x}\in{\cal X}^{\mathbb{N}}:x_0\in{A_0},x_1\in{A_1},\ldots,x_n\in{A_n}\big\},~A_i\in{\cal B}({\cal X }_i),~i=0,1,\ldots,n$, we define ${\bf P}(\cdot)$ on ${\cal B}({\cal X}^{\mathbb{N}})$ by
\begin{eqnarray}
{\bf P}(A)\tri\int_{A_0}p_0(dx_0)\ldots\int_{A_n}p_n(dx_n;x^{n-1})\equiv{\mu}_{0,n}(A_{0,n})\label{equation2}
\end{eqnarray}
where $A_{0,n}=\times_{i=0}^n{A_i}$, and ${\mu}_{0,n}(\cdot)$ denotes the restriction of the measure ${\bf P}(\cdot)$ on cylinder sets $A\in{\cal B}({\cal X}_{0,n})$, for $n\in\mathbb{N}$.\\
\noi{\bf Reproduction Distribution.} Consider the sequence of reproduction distributions $\{q_n(dy_n;y^{n-1},x^n):n\in\mathbb{N}\}$, where $q_n(\cdot;\cdot,\cdot)\in{\cal Q}({\cal Y}_n;{\cal Y}_{0,n-1}\times{\cal X}_{0,n})$. For a cylinder set $B\tri\big\{{\bf y}\in{\cal Y}^{\mathbb{N}}:y_0{\in}B_0,y_1{\in}B_1,\ldots,y_n{\in}B_n\big\}$, we define ${\bf Q}(\cdot|{\bf x})$ on ${\cal B}({\cal Y}^{\mathbb{N}})$ by
\begin{align}
{\bf Q}(B|{\bf x})
&\tri\int_{B_0}q_0(dy_0;x_0)\ldots\int_{B_n}q_n(dy_n;y^{n-1},x^n)\label{equation4}\\
&\equiv{\overrightarrow{Q}}_{0,n}(B_{0,n}|x^n),~B_{0,n}\in{\cal B}({\cal Y}_{0,n}).\label{equation4b}
\end{align}
\noi For Polish spaces, it can be shown \cite[Sec. II]{charalambous-stavrou2012isit} that any family of measures ${\bf Q}(\cdot|{\bf x})$ on ${\cal B}({\cal Y}^{\mathbb{N}})$ defined by (\ref{equation4}) is equivalent to a family of measures ${\bf Q}(\cdot|{\bf x})$ on $({\cal Y}^{\mathbb{N}},{\cal B}({\cal Y}^{\mathbb{N}}))$ satisfying the following consistency condition.\\ 
{\bf C1}: If $D\in{\cal B}({\cal Y}_{0,n}),$ then ${\bf Q}(D|{\bf x})$ is ${\cal B}({\cal X}_{0,n})-$measurable function of ${\bf x}\in{\cal X}^{\mathbb{N}}$.\\
\noi We denote the set of measures satisfying {\bf C1} by ${\cal Q}^{\bf C1}({\cal Y}^{\mathbb{N}};{\cal X}^{\mathbb{N}})\subseteq{\cal Q}({\cal Y}^{\mathbb{N}};{\cal X}^{\mathbb{N}})$.\\ 
Indeed, for any family of measures ${\bf Q}(\cdot|{\bf x})$ on $({\cal Y}^{\mathbb{N}},{\cal B}({\cal Y}^{\mathbb{N}}))$ satisfying consistency condition {\bf C1} one can construct a collection of probability distributions $\{q_n(dy_n;y^{n-1},x^n):n\in\mathbb{N}\}$ which are connected to ${\bf Q}(\cdot|{\bf x})$ via relation (\ref{equation4}) \cite[Sec. II]{charalambous-stavrou2012isit}. Here, $\overrightarrow{Q}_{0,n}(\cdot|x^n)\in{\cal Q}^{\bf C1}({\cal Y}_{0,n};{\cal X}_{0,n})$ denotes the restriction of ${\bf Q}(\cdot|{\bf x})\in{\cal Q}^{\bf C1}({\cal Y}^{\mathbb{N}};{\cal X}^{\mathbb{N}})$ to finite coordinates.
\par Next, we  introduce the precise definition of information NRDF by using relative entropy. Given  ${\bf P}(\cdot)\in{\cal M}({\cal X}^{\mathbb{N}})$ and  ${\bf Q}(\cdot|\cdot)\in{\cal Q}^{\bf C1}({\cal Y}^{\mathbb{N}};{\cal X}^{\mathbb{N}})$ we define {\it the joint distribution} on ${\cal X}^{\mathbb{N}}\times{\cal Y}^{\mathbb{N}}$ by $P_{0,n}(dx^n,dy^n)\tri({\mu}_{0,n}\otimes{\overrightarrow Q}_{0,n})(dx^n,dy^n)$, {\it the marginal distribution} on ${\cal Y}^{\mathbb{N}}$ by $\nu_{0,n}(dy^n)\tri({\mu}_{0,n}\otimes{\overrightarrow Q}_{0,n})({\cal X}_{0,n},dy^n)$, and {\it the product distribution} ${\overrightarrow\Pi}_{0,n}:{\cal B}({\cal X}_{0,n})\odot{\cal B}({\cal Y}_{0,n})\mapsto[0,1]$ by
\begin{align}
&{\overrightarrow\Pi}_{0,n}(dx^n,dy^n)\tri({\mu}_{0,n}\times\nu_{0,n})(dx^n,dy^n)\nonumber\\
&\tri\int_{A_0}p_0(dx_0)\ldots
\int_{A_n}p_n(x_n;x^{n-1})\int_{B_n}\nu_n(dy_n;y^{n-1}).\nonumber
\end{align}
\noi The information theoretic measure of interest is a special case of directed information \cite[Sec.~IV]{stavrou-charalambous2013isit} defined by relative entropy $\mathbb{D}(\cdot||\cdot)$
\begin{align}
&I_{\mu_{0,n}}(X^n\rightarrow{Y}^n)\tri\mathbb{D}({\mu}_{0,n} \otimes {\overrightarrow Q}_{0,n}||{\overrightarrow\Pi}_{0,n})\label{equation33}\\
&=\int \log \Big( \frac{{\overrightarrow Q}_{0,n}(d y^n|x^n)}{\nu_{0,n}(dy^n)}\Big)({\mu}_{0,n}\otimes {\overrightarrow Q}_{0,n})(dx^n,dy^n)\label{equation203}\\
&\equiv{\mathbb{I}}_{X^n\rightarrow{Y^n}}({\mu}_{0,n}, {\overrightarrow Q}_{0,n}).\label{equation7a}
\end{align}
\noi The notation ${\mathbb{I}}_{X^n\rightarrow{Y^n}}(\cdot,\cdot)$ indicates the functional dependence of $I_{\mu_{0,n}}(X^n\rightarrow{Y^n})$ on $\{{\mu}_{0,n}, {\overrightarrow Q}_{0,n}\}$. Consider a measurable distortion function $d_{0,n}(x^n,y^n):{\cal X}_{0,n}\times{\cal Y}_{0,n}\mapsto[0,\infty]$,~$d_{0,n}=\sum_{i=0}^n\rho(x_i,y_i)$, and define the fidelity of reproduction by
\begin{align}
&{\cal Q}^{\bf C1}_{0,n}(D)\tri\Big\{\overrightarrow{Q}_{0,n}(\cdot|x^n)\in{\cal Q}^{\bf C1}({\cal Y}_{0,n};{\cal X}_{0,n}):\nonumber\\
&\frac{1}{n+1}\int d_{0,n}({x^n},{y^n})(\mu_{0,n}\otimes\overrightarrow{Q}_{0,n})(d{x}^{n},d{y}^{n})\leq D\Big\},~D\geq0.\nonumber
\end{align}
\noi Next, we define the information NRDF.
\begin{definition}(Information NRDF) The information NRDF is 
\begin{align}
{R}^{na}_{0,n}(D) \tri  \inf_{{\overrightarrow{Q}_{0,n}(\cdot|x^n)\in{\cal Q}^{\bf C1}_{0,n}(D)}}\mathbb{I}_{X^n\rightarrow{Y^n}}(\mu_{0,n},{\overrightarrow Q}_{0,n}).
\label{ex12}
\end{align}
If the infimum over ${\cal Q}^{\bf C1}_{0,n}(D)$ in (\ref{ex12}) does not exist then we set ${R}^{na}_{0,n}(D)=\infty$. The information NRDF rate is  
\begin{align}
{R}^{na}(D)=\lim_{n\longrightarrow\infty}\frac{1}{n+1}{R}^{na}_{0,n}(D)\label{equation22}
\end{align}
provided the limit on the right hand side (RHS) of (\ref{equation22}) exists (if not we use $\limsup_{n\rightarrow\infty}$ ). If the infimum over ${\cal Q}^{\bf C1}_{0,n}(D)$ does not exist then we set ${R}^{na}(D)=\infty$.
\end{definition}
Note that $R_{0,n}^{na}(D)$ is also related to classical RDF \cite{berger}, denoted by $R_{0,n}(D)$, as follows. Let 
${\cal Q}_{0,n}(D)=\big\{P_{Y^n|X^n}(\cdot|x^n):\frac{1}{n+1}\int{d}_{0,n}(x^n,y^n){P}_{Y^n|X^n}(dy^n|x^n)\otimes{P}_{X^n}(dx^n)\leq{D}\big\},D\geq{0}$, then
\begin{align*}
&R_{0,n}(D)=\inf_{P_{Y^n|X^n}(\cdot|x^n)\in{\cal Q}_{0,n}(D)}\mathbb{D}(P_{Y^n|X^n}\otimes{P}_{X^n}||P_{Y^n}\times{P}_{X^n})\nonumber\\
&\stackrel{(a)}\leq\inf_{P_{Y^n|X^n}(\cdot|x^n)\in{\cal Q}_{0,n}(D)\cap{\cal Q}^{\bf C1}_{0,n}(D)}\mathbb{D}(P_{Y^n|X^n}\otimes{P}_{X^n}||P_{Y^n}\times{P}_{X^n}).
\end{align*}
\noi For memoryless sources, $(a)$ holds with equality.

\section{Optimization of NRDF and Properties}\label{existence}

\par In this section, we state conditions for the existence of solution to the extremum problem (\ref{ex12}), we give the optimal reproduction minimizing (\ref{ex12}) and some of its properties. These results are used when  we discuss the various applications.\\ 
\noi The following existence result is outlined in \cite{stavrou-charalambous2013isit}, while a complete derivation is given in \cite[Sec. III]{stavrou-kourtellaris-charalambous2014aaa}.
\begin{theorem}\cite[Sec. III]{stavrou-kourtellaris-charalambous2014aaa}(Existence)\label{existence_rd}
Suppose {\bf(A1)} ${\cal Y}_{0,n}$ is a compact; {\bf(A2)} for all $h(\cdot){\in}BC({\cal Y}_{n})$,  $(x^{n},y^{n-1})\in{\cal X}_{0,n}\times{\cal Y}_{0,n-1}\mapsto\int_{{\cal Y}_n}h(y)P_{Y|Y^{n-1},X^n}(dy|y^{n-1},x^n)\in\mathbb{R}$ 
is continuous jointly in  $(x^{n},y^{n-1})\in{\cal X}_{0,n}\times{\cal Y}_{0,n-1}$; {\bf(A3)} $d_{0,n}(x^n,\cdot)$ is continuous on ${\cal Y}_{0,n}$; {\bf(A4)} There exist  $(x^n,y^{n})\in{\cal X}_{0,n}\times{\cal Y}_{0,n}$ such that  $d_{0,n}(x^n,y^{n})<D$.\\
Then the infimum in ${R}^{na}_{0,n}(D)$ is achieved by some $\overrightarrow{Q}^*_{0,n}(dy^n|x^n)\in{\cal Q}^{\bf C1}_{0,n}(D)$. 
\end{theorem}
\noi It can be easily shown that $R_{0,n}^{na}(D)$ is equivalent to Gorbunov and Pinsker \cite{gorbunov-pinsker} definition of nonanticipatory $\epsilon$-entropy defined via mutual information $I(X^n;Y^n)$ by $R^{\epsilon}_{0,n}(D)=\inf\big{\{}I(X^n;Y^n):~P_{Y^n|X^n}(\cdot|x^n)\in{\cal Q}_{0,n}(D)\cap\{X_{i+1}^n\leftrightarrow{X}^i\leftrightarrow{Y}^i,~i=0,1,\ldots,n-1\}\big{\}}$. An extensive elaboration on the equality is given in \cite[Sec. III]{stavrou-kourtellaris-charalambous2014aaa}. By combining Theorem~\ref{existence_rd} and \cite[Th. 2-4]{gorbunov-pinsker} we have the following important results.
\begin{corollary}\label{corollary} Suppose the conditions of Theorem~\ref{existence_rd} hold. In addition, assume {\bf (A5)} the source is stationary; {\bf (A6)} for any $k=1,2,\ldots$, the sets ${\cal Q}_{0,n}^{\bf C1}(D)$ and ${\cal Q}^{\bf C1}_{k,n+k}(D)$ are copies of the same set.\\
Then $\lim_{n\rightarrow\infty}\frac{1}{n+1}{R}^{na}_{0,n}(D)$ exists and it is finite.\\
If also, {\bf(A7)} $\overrightarrow{Q}_{0,n}(\cdot|x^n)\in{\cal Q}_{0,n}^{\bf C1}(D)$ implies $\overrightarrow{Q}_{0,k}(\cdot|x^k)\in{\cal Q}_{0,k}^{\bf C1}(D)$, $\overrightarrow{Q}_{k+1,n}(\cdot|x_{k+1},\ldots,x_n)\in{\cal Q}_{k+1,n}^{\bf C1}(D)$ $\forall{k}=0,1,\ldots,n-1$;~{\bf(A8)} for any $\alpha_t:[s_1,s_2]\longmapsto[0,\infty)$, $\sum_{t=s_1}^{s_2}\alpha_t=1$, $\forall{t}\in(0,\infty)$, and $P_{Y^n|X^n}(\cdot|x^n)\in{\cal Q}_{0,n}(D)\Longrightarrow{P}_{\tilde{Y}^n|X^n}(\cdot|x^n)\in{\cal Q}_{0,n}(D)$, $\forall{n}\in\mathbb{N}$, where ${P}_{X^n,\tilde{Y}^n}({\cal A})=({P}_{\tilde{Y}^n|X^n}\otimes{P}_{X^n})({\cal A})\tri\sum_{s_1}^{s_2}\alpha_s({P}_{\tilde{Y}^n|X^n}\otimes{P}_{X^n})({\cal A}_s)$, ${\cal A}\tri\{(X_i,Y_i)=(x_i,y_i):~i=0,1,\ldots\}\subseteq{\cal X}^{\mathbb{N}}\times{\cal Y}^{\mathbb{N}}$, ${\cal A}_s\tri\{(X_{i-s},Y_{i-s})=(x_i,y_i):~i=0,1,\ldots\}$.\\
Then the infimum in (\ref{ex12}) is achieved by $\overrightarrow{Q}^*_{0,n}(dy^n|x^n)\in{\cal Q}^{\bf C1}_{0,n}(D)$ and $\{(X_n,Y_n):~n\in\mathbb{N}\}$ is jointly stationary.
\end{corollary}
\begin{proof}
By Theorem~\ref{existence_rd}, $R_{0,n}^{na}(D)$ is finite for any finite $n$. Using this and \cite[Theorem 2-4]{gorbunov-pinsker}, the results follow.
\end{proof}
\noi Utilizing the convexity of the extremum problem (\ref{ex12}) (see \cite[Th. II.2]{stavrou-kourtellaris-charalambous2014aaa}), and applying variational methods, the general closed form expression of the optimal stationary reproduction conditional distribution of (\ref{ex12}) is derived in \cite[Sec. IV]{charalambous-stavrou-ahmed2014a}. Here, we only state the main theorem. 
\begin{theorem}\cite{charalambous-stavrou-ahmed2014a}({Optimal stationary reproduction distribution})\label{th6}
We suppose the optimal reproduction distribution and source distribution are stationary, i.e., conditions of Corollary~\ref{corollary} hold. The optimal solution of information NRDF is given by\footnote{Due to stationarity assumption $\nu^*_i(\cdot;\cdot)=\nu^*(\cdot;\cdot)$ and ${q}^*_{i}(\cdot;\cdot,\cdot)={q}^*(\cdot;\cdot,\cdot)$.}
\begin{align}
&\overrightarrow{Q}^*_{0,n}(dy^n|x^n)=\otimes_{i=0}^n{q}^*_{i}(dy_i;y^{i-1},x^i)\nonumber\\
&=\otimes_{i=0}^n\frac{e^{s \rho(x_i,y_i)}\nu^*_{i}(dy_i;y^{i-1})}{\int_{{\cal Y}_i} e^{s \rho(x_i,y_i)} \nu^*_{i}(dy_i;y^{i-1})},~s\leq{0}\label{ex14}
\end{align}
and $\nu^*_{i}(\cdot;y^{i-1})\in {\cal Q}({\cal Y}_i;{\cal Y}_{0,{i-1}})$. The information NRDF is given by
\begin{align}
&{R}^{na}_{0,n}(D)=sD(n+1)-\sum_{i=0}^n\int\log \Big( \int e^{s\rho(x_i,y_i)}\nu^*_{i}(dy_i;y^{i-1})\Big)\nonumber\\
&\quad\times{\overrightarrow{Q}^*_{0,i-1}(dy^{i-1}|x^{i-1})\otimes{\mu}_{0,i}(dx^i).}\nonumber
\end{align}
Moreover, if ${R}^{na}_{0,n}(D) > 0$ then $ s < 0$, and
$$\frac{1}{n+1}\sum_{i=0}^n\int\rho(x_i,y_i)\overrightarrow{Q}^*_{0,i}(dy^i|x^i)\otimes{\mu}_{0,i}(dx^i)=D.$$
\end{theorem}
\begin{remark}\label{markov_stationary} Note that for single letter distortion function $\rho(x_i,y_i)$ the optimal reproduction is Markov with respect to $x_i$ given by $q_i^*(dy_i;y^{i-1},x_i)$. If the distortion function is generalized to $\rho(x_i,T^i{y^n})$, where $T^i{y^n}$ is the shift operator on $y^n$, then
$\overrightarrow{Q}_{0,n}^*(dy^n|x^n)$ is given by (\ref{ex14}) with $\rho(x_i,y_i)$ replaced by $\rho(x_i,T^i{y^n})$, and similarly for $\rho(T^ix^n,y_i)$. 
\end{remark}
\noi Next, we present an alternative equivalent characterization of the solution of $R_{0,n}^{na}(D)$, which can be used to derive a lower bound on $R_{0,n}^{na}(D)$ similar to the SLB \cite[Ch. 4]{berger}. 
\begin{theorem}(Alternative characterization)\label{alternative_expression} Suppose the conditions of Theorem~\ref{th6} hold. Then 
\begin{align*}
R^{na}_{0,n}(D)&=\max_{s\leq{0}}\max_{\lambda\in\Psi_s}\big\{sD(n+1)+\sum_{i=0}^n\int\log\big(\lambda_i(x^i,y^{i-1})\big)\nonumber\\
&\times{P}_{0,i-1}(dx^{i-1},dy^{i-1})\otimes{p}_{i}(dx_i;x^{i-1})\big\}
\end{align*}
where
$\Psi_s\tri\big\{\lambda\tri\{\lambda_i(x^i,y^{i-1})\geq{0}:~i=0,1,\ldots,n\}:\int{e}^{s\rho(x_i,y_i)}\lambda_i(x^i,y^{i-1})P_{0,i}(dx^{i}|y^{i-1})\leq{1},~i=0,1,\ldots,n\big\}$.
\end{theorem}
\begin{proof}
The derivation is found in \cite[App. E]{stavrou-kourtellaris-charalambous2014aaa}.
\end{proof}

\section{Applications via Examples}
\par In this section, we describe some applications of information NRDF using the following two working examples: {\bf(i)}  the BSMS($p$), {\bf(ii)} the multidimensional Gaussian stationary source.
\vspace*{0.2cm}\\
\noi{\bf Bound and RL due to Causal Codes.} Let $R(D)$ denotes the OPTA by noncausal codes \cite{berger}, and $r^c(D)$ the OPTA by causal codes \cite{neuhoff1982}. Then we have the following bounds.
\begin{align}
R(D)\stackrel{(b)}\leq{R}^{na}(D)\stackrel{(c)}\leq{r}^c(D)\label{equationX1}
\end{align}
where $(b)$ follows from the fact that $R(D)$ is optimized over a larger set than that of $R^{na}(D)$, and $(c)$ follows by the converse coding theorem and \cite{neuhoff1982}. Since the OPTA by noncausal codes for sources with memory is often unknown (unless one consider memoryless or Gaussian sources), then $R^{na}(D)$ can be used to find an upper bound to the OPTA by noncausal codes. For memoryless sources $R(D)=R^{na}(D)$, and this bound is tight. Moreover, since $r^c(D)-R(D)\geq{R}^{na}(D)-R(D)$, we can find the RL of causal codes with respect to the noncausal codes using $R^{na}(D)$.\\ 
\noi{\bf Noisy Coding Theorem (Source-Channel Matching).} An operational definition for $R^{na}(D)$ can be established by using symbol-by-symbol transmission, provided for a given source and distortion function we can find the optimal reproduction distribution, and then realize it over an encoder-channel-decoder, so that the source is matched to the channel. We give an example for multidimensional Gaussian stationary sources providing a noisy coding theorem for $R^{na}(D)$.

\subsection{BSMS(p): Exact Solution, Bounds, and Rate Loss}\label{example:bsms}

\par Consider a BSMS($p$), with stationary transition probabilities
$\big\{P_{X_i|X_{i-1}}(x_i|x_{i-1}):~(x_i,x_{i-1})\in\{0,1\}\times\{0,1\}\big\}$ given by $P_{X_i|X_{i-1}}(0|0)=P_{X_i|X_{i-1}}(1|1)=1-p$, $P_{X_i|X_{i-1}}(1|0)=P_{X_i|X_{i-1}}(0|1)=p$, $i\in 0,1,\ldots$, and  single letter Hamming distortion criterion, $\rho(x,y)=0$ if $x=y$ and $\rho(x,y)=1$ if $x \neq y$. The solution to the NRDF is given to the next theorem. 
\begin{theorem}\label{marex1} For a BSMS($p$) and single letter Hamming distortion
\begin{eqnarray}
{ R}^{na}(D) = \left\{ \begin{array}{ll}
         H(m)-H(D) & \mbox{if $D \leq \frac{1}{2}$} \\
        0 & \mbox{otherwise}\end{array} \right. \label{equationA}
\end{eqnarray}
where $m=1-p-D+2pD$, and the optimal (stationary) reproduction distribution is 
\begin{align}
P_{Y_i|X_i,Y_{i-1}}^*(y_i|x_i,y_{i-1})=\bbordermatrix{~ &  &  &    \cr
                   & \alpha & \beta& 1-\beta & 1-\alpha\vspace{0.3cm} \cr
                   & 1-\alpha & 1-\beta& \beta &  \alpha \cr}\nonumber
\end{align}
where $\alpha=\frac{(1-p)(1-D)}{1-p-D+2pD},~\beta=\frac{p(1-D)}{p+D-2pD}$.
\end{theorem}
\begin{proof}
The proof is found in \cite[Th.~IV.11]{stavrou-kourtellaris-charalambous2014aaa}. 
\end{proof}
\noi Note that for $p=\frac{1}{2}$, then BSMS($\frac{1}{2}$) is the IID Bernoulli source, and $R^{na}(D)=1-H(D)\equiv{R}(D)$, $D<\frac{1}{2}$, as expected.
\begin{figure}[ht]
\centering
\includegraphics[scale=0.5]{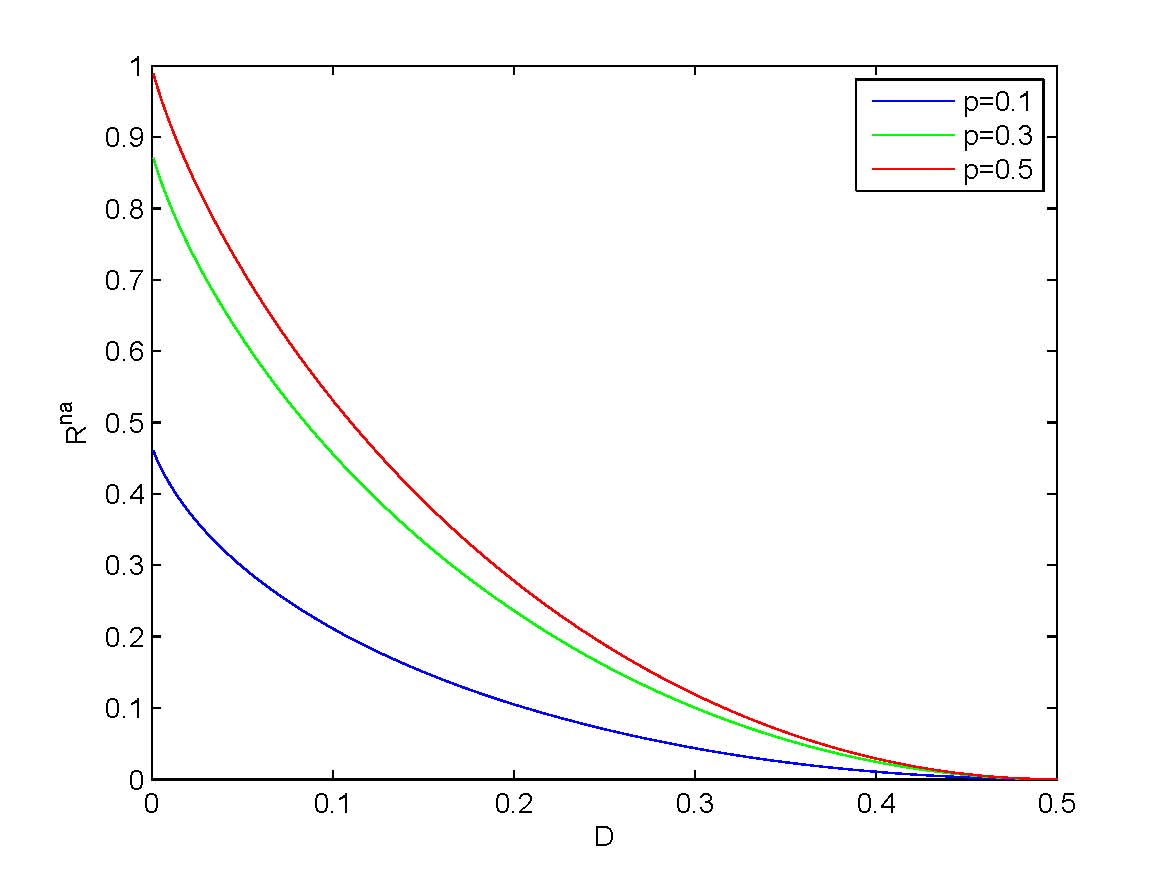}
\caption{$R^{na}(D)$ for different values of parameter $p$.}
\label{nardf-graph}
\end{figure}
\noi The graph of $R^{na}(D)$ is illustrated in Fig.~\ref{nardf-graph}.\\
\noi{\bf Bounds on $\bf{R(D)}$.} The classical RDF for the BSMS($p$) is only known for the distortion region $0\leq{D}\leq{D}_c$ \cite{gray1971}, while for the rest distortion region only bounds are known \cite{berger1977}. Fig.~\ref{figmakrov} shows the graph of $R(D)$ for $0\leq{D}\leq{D}_c$, Berger's lower and upper bounds \cite{berger1977}, SLB, and the upper bound based on $R^{na}(D)$. We observe that for $p=0.25$, the upper bound based on $R^{na}(D)$ does slightly better than Berger's upper bound. However, for small values of $D$, we have observed 
\begin{figure}[ht]
\centering
\includegraphics[scale=0.5]{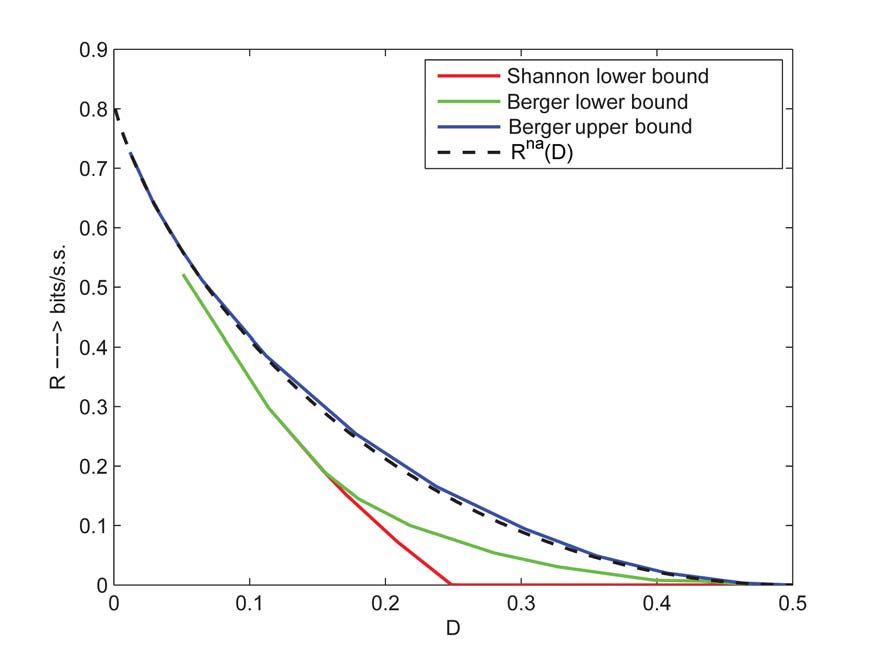}
\caption{$R(D)$ for BSMS($p$) for $0\leq{D}\leq{D}_c$ and Bounds for $p=0.25$.}
\label{figmakrov}
\end{figure}
that Berger's upper bound fails to be tight, while the one based on $R^{na}(D)$ is tight \cite[Sec. V.C]{stavrou-kourtellaris-charalambous2014aaa}.\\
\noi{\bf RL of Causal Codes.} By utilizing the bound $R^{na}(D)\geq{R}(D)$, we can deduce that the $RL$ of causal codes for the BSMS($p$) cannot exceed $R^{na}(D)-R(D)=H(m)-H(q),~0\leq{D}\leq{D}_c$, where $R(D)=H(q)-H(D)$, $p=1-q,~q\leq\frac{1}{2},~D\leq{D}_c=\frac{1}{2}\Big(1-\sqrt{1-\big(\frac{q}{p}\big)^2}\Big)$.  Note that the exact value of $RL$ is only given for the region $0\leq{D}\leq{D}_c$, where the exact solution of $R(D)$ is known. Beyond this region, upper and lower bounds for $RL$ can be found \cite[Sec. V.C]{stavrou-kourtellaris-charalambous2014aaa}.

\subsection{Multidimensional Gaussian Stationary Sources: Source-Channel Matching and Rate Loss}

\par In this section, we consider a vector partially observable Gaussian-Markov process and we compute explicitly the closed form expression of $R^{na}(D)$. This expression makes feasible the matching of the source to the channel.\\
\noi Consider the following multidimensional partially observed linear Gauss-Markov system
\begin{eqnarray}
\left\{ \begin{array}{ll} Z_{t+1}=AZ_t+BW_t,~Z_0=z,~t\in\mathbb{N}\\
X_t=CZ_t+NV_t,~t\in\mathbb{N} \end{array} \right.\label{equation51}
\end{eqnarray}
where $Z_t\in\mathbb{R}^m$ is the state (unobserved) process and $X_t\in\mathbb{R}^p$ is the information source, obtained from noisy measurements of $CZ_t$. In this application the objective is to compress the sensor data, which is the only observable information. Next, we introduce certain assumptions which are standard in infinite horizon Kalman Filter \cite{caines1988}, and they are also sufficient for existence of the limit, $\lim_{n\longrightarrow\infty}\frac{1}{n+1}R^{na}_{0,n}(D)$.\\
\noi{\bf (E1)} ($C,A$) is detectable and ($A,\sqrt{BB^{tr}}$) is stabilizable, ($N\neq0$); {\bf(E2)} the state and observation noise $\{(W_t,V_t):t\in\mathbb{N}\}$ are Gaussian IID vectors $W_t\in\mathbb{R}^k$, $V_t\in\mathbb{R}^d$, mutually independent with parameters $N(0,I_{k\times{k}})$ and $N(0,I_{d\times{d}})$, independent of the Gaussian RV $Z_0$, with parameters $N(\bar{z}_0,\bar{\Sigma}_0)$; {\bf (E3)} the distortion function is single letter defined by $d_{0,n}(x^n,{y}^n)\tri\sum_{t=0}^n||x_t-{y}_t||_{\mathbb{R}^p}^2$.\\ 
\noi According to Theorem~\ref{th6}, the optimal stationary reproduction distribution is given for $s\leq{0}$ by
\begin{align}
{P}^*_{{Y}_t|Y^{t-1},X_t}(d{y}_t|y^{t-1},x_t)=\frac{e^{s||{y}_t-x_t||_{\mathbb{R}^p}^2}P^*_{{Y}_t|{Y}^{t-1}}(d{y}_t|{y}^{t-1})}{\int_{{\cal Y}_t}e^{s||{y}_t-x_t||_{\mathbb{R}^p}^2}P^*_{{Y}_t|{Y}^{t-1}}(d{y}_t|{y}^{t-1})}.\label{eq.9}
\end{align}
Note that the exponential quadratic term in (\ref{eq.9}) implies that $P_{{Y}_t|{Y}^{t-1},X_t}(\cdot|{y}^{t-1},x_t)$ is conditionally Gaussian (using completion of squares if necessary). Hence, the channel connecting $\{X_t:t\in\mathbb{N}\}$ to $\{{Y}_t:t\in\mathbb{N}\}$ has the general form
\begin{eqnarray}
{Y}_t=\bar{A}X_t+\bar{B}{Y}^{t-1}+V^c_t,~t\in\mathbb{N}\label{eq.10}
\end{eqnarray}
where $\bar{A}\in\mathbb{R}^{p\times{p}}$, $\bar{B}\in\mathbb{R}^{p\times{t}p}$, and $\{V^c_t:~t\in\mathbb{N}\}$ is an independent sequence of Gaussian vectors with zero mean and covariance $cov(V^c_t)=Q=diag\{q_1,\ldots,q_p\}$.  Consider a pre-encoder introducing the Gaussian error process $\{K_t:~t\in\mathbb{N}\}$, $K_t\triangleq{X}_t-\mathbb{E}\{X_t|{Y}^{t-1}\}$ and its steady state covariance $\Lambda_\infty$, $\Lambda_\infty=\lim_{n\rightarrow\infty}\Lambda_t$,~$\Lambda_t\triangleq\mathbb{E}\{K_tK_t^{tr}\},~t\in\mathbb{N}$. Let $E_\infty$ be a unitary matrix such that 
\begin{eqnarray}
E_\infty\Lambda_\infty{E}_\infty^{tr}=diag\{\lambda_{\infty,1},\ldots\lambda_{\infty,p}\},~\Gamma_t\triangleq{E}_\infty{K}_t,~t\in\mathbb{N}.\label{equation53}
\end{eqnarray}
\noi Analogously, introduce the process $\{\tilde{K}_t:~t\in\mathbb{N}\}$ defined by $\tilde{K}_t\tri{Y}_t-\mathbb{E}\{X_t|{Y}^{t-1}\}\equiv{Y}_t-\widehat{X}_{t|t-1}$,~$\tilde{\Gamma}_t=E_\infty\tilde{K}_t$. It is easily shown that $d_{0,n}(X^n,{Y}^n)=d_{0,n}(K^n,\tilde{K}^n)=\sum_{t=0}^n||\tilde{K}_t-K_t||_{\mathbb{R}^p}^2=\sum_{t=0}^n||\tilde{\Gamma}_t-\Gamma_t||_{\mathbb{R}^p}^2$.
\begin{figure}[ht]
\centering
\includegraphics[scale=0.5]{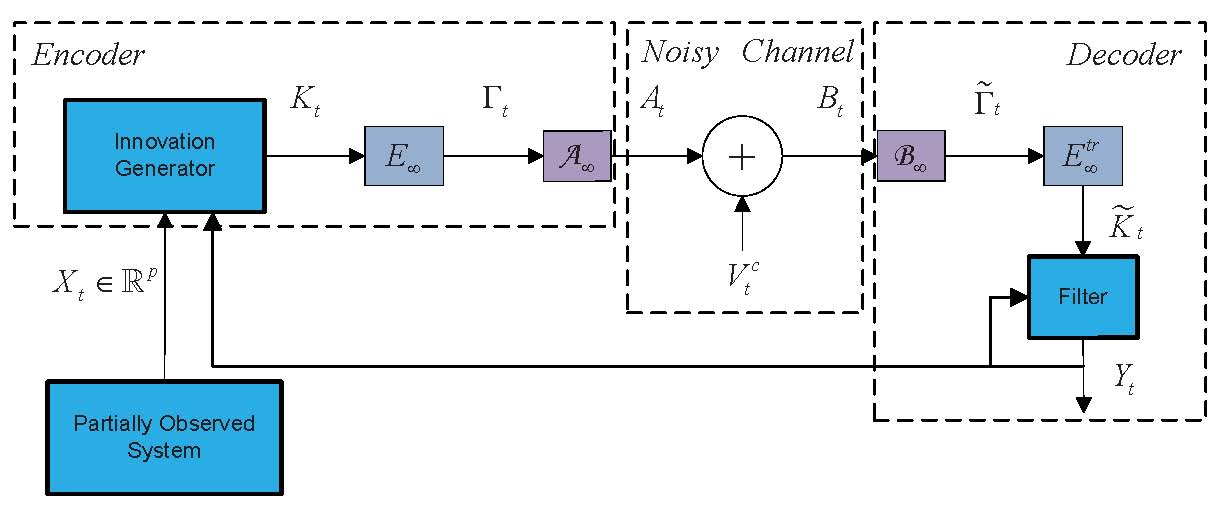}
\caption{Realization of the optimal stationary reproduction distribution.}
\label{discrete_time_communication_system}
\end{figure}
\noi Using basic properties of conditional entropy we can show that $R^{na}(D)=\lim_{n\longrightarrow\infty}\frac{1}{n+1}R_{0,n}^{na,K^n,\tilde{K}^n}(D)=\lim_{n\longrightarrow\infty}\frac{1}{n+1}R_{0,n}^{na,\Gamma^n,\tilde{\Gamma}^n}(D)$. Next, we state the main result.
\begin{theorem}\label{solution_gaussian}
Under Assumptions ({\bf E1})-({\bf E3}), the information NRDF rate for (\ref{equation51}) is given by
\begin{align}
R^{na}(D)=\frac{1}{2}\sum_{i=1}^p\log\Big{(}\frac{\lambda_{\infty,i}}{\delta_{\infty,i}}\Big{)}\nonumber
\end{align}
where $diag\{\lambda_{\infty,1},\ldots,\lambda_{\infty,p}\}=\lim_{t\longrightarrow\infty}E_t\Lambda_tE^{tr}_{t}=E_\infty{\Lambda}_\infty{E}_\infty^{tr}$,~$\Lambda_\infty=C\Sigma_{\infty}C^{tr}+NN^{tr}$
\begin{eqnarray}
\delta_{\infty,i} \tri\left\{ \begin{array}{ll} \xi_\infty & \mbox{if} \quad \xi_\infty\leq\lambda_{\infty,i} \\
\lambda_{\infty,i} &  \mbox{if}\quad\xi_\infty>\lambda_{\infty,i} \end{array} \right.,~i=2,\ldots,p\nonumber
\end{eqnarray}
and $\xi_\infty$ is chosen such that $\sum_{i=1}^p\delta_{\infty,i}=D$. Define $H_\infty=\lim_{t\longrightarrow\infty}H_t,~H_t\tri{d}iag\{\eta_{t,1},\ldots,\eta_{t,p}\},~\eta_{t,i}=1-\frac{\delta_{t,i}}{\lambda_{t,i}},~i=1,\ldots,p,~{\cal B}_{\infty}=\lim_{t\longrightarrow\infty}{\cal B}_t=\sqrt{H_\infty\Delta_\infty{Q}^{-1}},~{\cal B}_t\triangleq\sqrt{H_t\Delta_tQ^{-1}},~\Delta_\infty=\lim_{t\longrightarrow\infty}\Delta_t,~\Delta_t=diag\{\delta_{t,1},\ldots,\delta_{t,p}\}, t\in\mathbb{N}$. Moreover, $\Sigma_\infty$ is the steady state covariance of the error $Z_t-\mathbb{E}\{Z_t|{Y}^{t-1}\}\sim{N}(0,\Sigma_\infty)$ of the Kalman filter given by
\begin{align}
&\widehat{Z}_{t+1|t}=A\widehat{Z}_{t|t-1}\nonumber\\
&+A\Sigma_\infty(E_\infty^{tr}H_{\infty}E_{\infty}C)^{tr}M_\infty^{-1}\big({Y}_t-C\widehat{Z}_{t|t-1}\big)\nonumber\\
&\Sigma_{\infty}=A\Sigma_\infty{A}^{tr}\nonumber\\
&-A\Sigma_{\infty}(E_\infty^{tr}H_\infty{E}_{\infty}C)^{tr}M_{\infty}^{-1}(E_{\infty}^{tr}H_{\infty}E_{\infty}C)\Sigma_{\infty}A^{tr}+BB_{\infty}^{tr}\nonumber\\
&M_\infty=E_\infty^{tr}H_\infty{E}_{\infty}C\Sigma_{\infty}(E_{\infty}^{tr}H_{\infty}E_{\infty}C)^{tr}\nonumber\\
&+E_{\infty}^{tr}H_{\infty}E_{\infty}NN^{tr}(E_{\infty}^{tr}H_{\infty}E_{\infty})^{tr}+E_{\infty}^{tr}{\cal B}_{\infty}Q{\cal B}_{\infty}^{tr}E_\infty\nonumber
\end{align}
where $\widehat{Z}_{t|t-1}\tri\mathbb{E}\{Z_t|{Y}^{t-1}\}$ and $\hat{Z}_0 =\mathbb{E}\{ Z_0| Y^{-1}\}, Z_0-\hat{Z}_0\sim{N}(0,\Sigma_\infty)$.
\end{theorem}
\begin{proof}
The proof is found in \cite[App. F]{stavrou-kourtellaris-charalambous2014aaa}.
\end{proof}
\noi{\bf Source-Channel Matching.} In view of Fig.~\ref{discrete_time_communication_system}, the conditional distribution of NRDF is realized via an encoder-channel-decoder. Moreover, the channel consists of parallel additive Gaussian noisy channels with feedback defined by
\begin{align*}
B_{t,i}=A_{t,i}(X_t,B^{t-1})+V_{t,i}^{c},~t\in\mathbb{N},~i=1,\ldots,p.
\end{align*}
Recall that the capacity of a parallel memoryless Gaussian channel with feedback subject to a power constraint $\frac{1}{n+1}\mathbb{E}\{\sum_{t=0}^n||A_t||_{\mathbb{R}^p}\leq{P}\}$, is given by $C(P)=\lim_{n\rightarrow\infty}\frac{1}{2}\frac{1}{n+1}\sum_{t=0}^n\sum_{i=1}^p\log|1+\mathbb{E}\{(A_{t,i})^2\}Q^{-1}|=\frac{1}{2}\sum_{i=1}^p\log(1+\frac{P_{\infty,i}}{q_i})$, $\sum_{i=0}^p{P}_{\infty,i}=P$, ${P}_{\infty,i}=\lim_{n\rightarrow\infty}{E}\{(A_{t,i})^2\}$. As a result, for a given $D\geq{0}$, we can let $P=D$, i.e., $\frac{P_{\infty,i}}{q_i}=\frac{\lambda_{\infty,i}}{\delta_{\infty,i}}-1$, then $C(P)=R^{na}(D)$, and the end-to-end distortion is satisfied.\\
\noi{\bf RL of Zero-Delay Codes.} The source distribution $\{X_t:~t\in\mathbb{N}\}$ in (\ref{equation51}) is Gaussian, hence we can compute the OPTA by noncausal codes, $R(D)$, by using power spectral density expression \cite{berger}. The RL of causal and zero-delay codes with respect to the noncausal codes is precisely $\frac{1}{2}\sum_{i=1}^p\log\Big{(}\frac{\lambda_{\infty,i}}{\delta_{\infty,i}}\Big{)}-R(D)$ bits/sample. 
\section*{Acknowledgement}
This work was financially supported by a medium size University of Cyprus
grant entitled ``DIMITRIS"  and by QNRF, a member of Qatar Foundation, under the project NPRP 6-784-2-329.

\bibliographystyle{IEEEtran}
\bibliography{photis_references_isit2014}

\end{document}